%% file: FDFA.tex
\documentclass[usenames,dvipsnames,svgnames,table]{CSML}
\pdfoutput=1

\usepackage{lastpage}
\newcommand{\dOi}{14(1:15)2018}
\lmcsheading{}{1--\pageref{LastPage}}{}{}%
{Dec.~29, 2016}{Feb.~14, 2018}{}

\usepackage{graphicx}
\usepackage{amssymb,amsmath,latexsym}
\usepackage{xspace}
\usepackage{stmaryrd}
\usepackage{tikz}
\usetikzlibrary{arrows,positioning,shapes,decorations,automata,backgrounds,petri,decorations.pathreplacing,calc}
\usepackage{hyperref}
\hypersetup{hidelinks}

\input{macros}

\begin{document}

\title{Families of DFAs as Acceptors of $\omega$-Regular Languages\rsuper{*}}
\titlecomment{{\lsuper*}\kern-1.5ptThe present article extends \cite{ABF16}.}

\author[D.~Angluin]{Dana Angluin}	%required
\address{Yale University, New Haven, CT, USA}	%required
%\email{}  %optional
%\thanks{}	%optional

\author[U.~Boker]{Udi Boker}	
\address{Interdisciplinary Center (IDC), Herzliya, Israel}	
%\email{}  %optional

\author[D.~Fisman]{Dana Fisman}	%optional
\address{Ben-Gurion University, Beer-Sheva, Israel}	%optional
%\email{}  %optional
%\urladdr{}  %optional
\thanks{This research was supported by the United States - Israel Binational Science Foundation (BSF) grant 2016239 and the Israel Science Foundation (ISF) grant 1373/16.}	%optional

%% The [ ] are required for running head on odd and even pages, use suitable abbreviations in case of long titles and many authors:

%%%%%%%%%%%%%%%%%%%%%%%%%%%%%%%%%%%%%%%%%%%%%%%%%%%%%%%%%%%%%%%%%%%%%%%%%%%

%% the abstract has to PRECEDE the command \maketitle: be sure not to issue the \maketitle command twice!

\begin{abstract}
Families of \dfa s (\fdfa s) provide an alternative formalism for recognizing $\omega$-regular languages. The motivation for introducing them was a desired correlation between the automaton states and right congruence relations, in a manner similar to the Myhill-Nerode theorem for regular languages. This correlation is beneficial for learning algorithms, and indeed it was recently shown that $\omega$-regular languages can be learned from membership and equivalence queries, using \fdfa s as the acceptors.

In this paper, we look into the question of how suitable \fdfa s are for defining $\omega$-regular languages. Specifically, we look into the complexity of performing Boolean operations, such as complementation and intersection, on \fdfa s, the complexity of solving decision problems, such as emptiness and language containment, and the succinctness of \fdfa s compared to standard deterministic and nondeterministic $\omega$-automata.

We show that \fdfa s enjoy the benefits of deterministic automata with respect to Boolean operations and decision problems. Namely, they can all be performed in nondeterministic logarithmic space. 
We provide polynomial translations of deterministic \buchi\ and \cobuchi\ automata to \fdfa s and of \fdfa s to nondeterministic \buchi\ automata (\nba s). We show that translation of an \nba\ to an \fdfa\ may involve an exponential blowup. Last, we show that \fdfa s are more succinct than deterministic parity automata  (\dpa s) in the sense that translating a \dpa\ to an \fdfa\ can always be done with only a polynomial increase, yet the other direction involves an inevitable exponential blowup in the worst case.
\end{abstract}

\maketitle

%% start the paper here:
\section{Introduction}
The theory of finite-state automata processing infinite words was developed in the early sixties, starting with \buchi~\cite{Buc60} and Muller~\cite{Mul63}, and motivated by problems in logic and switching theory. Today, automata for infinite words are extensively used in verification and synthesis of \emph{reactive systems}, such as operating systems and communication protocols.

An automaton processing finite words makes its decision according to the last visited state. On infinite words, \buchi\ defined that a run is accepting if it visits a designated set of states infinitely often. Since then several other accepting conditions were defined, giving rise to various $\omega$-automata, among which are Muller, Rabin, Streett and parity automata.

The theory of $\omega$-regular languages is more involved than that of finite words. This was first evidenced by \buchi's observation that nondeterministic \buchi\ automata are more expressive than their deterministic counterpart. While for some types of $\omega$-automata the nondeterministic and deterministic variants have the same expressive power, none of them possesses all the nice qualities of acceptors for finite words. In particular, none has a corresponding Myhill-Nerode theorem~\cite{N58}, i.e. a direct correlation between the states of the automaton and the equivalence classes corresponding to the canonical right congruence of the recognized language. 

The absence of a Myhill-Nerode like property in $\omega$-automata has been a major drawback in obtaining learning algorithms for $\omega$-regular languages, a question that has received much attention lately due to applications in verification and synthesis, such as black-box checking~\cite{PVY02}, assume-guarantee reasoning~\cite{AlurCMN05,PasareanuGBCB08, ChenCFTTW10, FengKP11}, error localization~\cite{WeimerN05,CCKKST15}, regular model checking~\cite{VardhanSVA05,NeiderJ13}, finding security bugs~\cite{RaffeltMSM09,ChaluparPPR14,FBJV14}, programming networks~\cite{RaffeltMSM09,YuanAL14} and more. The reason is that learning algorithms typically build on this correspondence between the automaton and the right congruence. 

Recently, several algorithms for learning an unknown $\omega$-regular language were proposed, all using non-conventional acceptors. One uses a reduction due to~\cite{CalbrixNP93} named $\ldollar$-automata of $\omega$-regular languages to regular languages~\cite{FarzanCCTW08}, and the others use a representation termed \emph{families of \dfa s} ~\cite{AngluinF14,LCZL16}. 
Both representations are founded on the following well known property of $\omega$-regular languages: two $\omega$-regular languages are equivalent iff they agree on the set of ultimately periodic words. An ultimately periodic word $uv^\omega$, where $u\in\Sigma^*$ and $v\in\Sigma^+$, can be represented as a pair of finite words $(u,v)$. Both $\ldollar$-automata and families of \dfa s process such pairs and interpret them as the corresponding ultimately periodic words. Families of \dfa s have been shown to be up to exponentially more succinct than  $\ldollar$-automata \cite{AngluinF14}.

A family of \dfa s (\fdfa) is composed of a \emph{leading automaton} $\Q$ with no accepting states and for each state $q$ of $\Q$, a \emph{progress \dfa} $\P_q$. Intuitively, the leading automaton is responsible for processing the non-periodic part $u$, and depending on the state $q$ reached when $\Q$ terminated processing $u$, the respective progress \dfa\ $\P_q$ processes the periodic part $v$, and determines whether the pair $(u,v)$, which  corresponds to $uv^\omega$, is accepted. (The exact definition is more subtle and is provided in Section~\ref{sec:fdfas}.) If the leading automaton has $n$ states and the size of the maximal progress \dfa\ is $k$, we say that the \fdfa\ is of size $(n,k)$. 
An earlier definition of \fdfa s, given in~\cite{Klarlund94}, provided a machine model for the \emph{families of right congruences} of \cite{MalerStaiger97}.\footnote{Another related formalism is of Wilke algebras \cite{Wil91,Wil93}, which are two-sorted algebras equipped with several operations. An $\omega$-language over $\Sigma^\omega$ is $\omega$-regular if and only if there exists a two-sorted morphism from $\Sigma^\infty$ into a finite Wilke structure \cite{Wil91}. A central difference between the FDFA theory and the algebraic theory of recognition by monoids, semigroups, $\omega$-semigroups, and Wilke structures is that the former relates to right-congruences, while the latter is based on two-sided congruences.} They were redefined in \cite{AngluinF14}, where their acceptance criterion was adjusted, and their size was reduced by up to a quadratic factor. We follow the definition of~\cite{AngluinF14}. 

In order for an \fdfa\ to properly characterize an $\omega$-regular language, it must satisfy the \emph{saturation} property: considering two pairs $(u,v)$ and $(u',v')$, if $uv^\omega=u'v'^\omega$ then either both $(u,v)$ and $(u',v')$ are accepted or both are rejected (cf.\ \cite{CalbrixNP93,Saec90}).
Saturated \fdfa s are shown to exactly characterize the set of $\omega$-regular languages.
Saturation is a semantic property, and the check of whether a given \fdfa\ is saturated is shown to be in PSPACE. Luckily, the \fdfa s that result from the learning algorithm of \cite{AngluinF14} are guaranteed to be saturated.

Saturated \fdfa s bring an interesting potential -- they have a Myhill-Nerode like property, and while they are ``mostly'' deterministic, a nondeterministic aspect is hidden in the separation of the prefix and period parts of an ultimately periodic infinite word. This gives rise to the natural questions of how ``dominant'' are the determinism and nondeterminism in \fdfa s, and how ``good'' are they for representing $\omega$-regular languages. 
These abstract questions translate to concrete questions that concern the succinctness of \fdfa s and the complexity of solving their decision problems, as these measures play a key role in the usefulness of applications built on top of them. 

Our purpose in this paper is to analyze the \fdfa\ formalism and answer these questions. Specifically, we ask: What is the complexity of performing the Boolean operations of complementation, union, and intersection on saturated \fdfa s? What is the complexity of solving the decision problems of membership, emptiness, universality, equality, and language containment for saturated \fdfa s? How succinct are saturated \fdfa s, compared to deterministic and nondeterministic $\omega$-automata?

We show that saturated \fdfa s enjoy the benefits of deterministic automata with respect to Boolean operations and decision functions. Namely, the Boolean operations can be performed in logarithmic space, and the decision problems can be solved in nondeterministic logarithmic space. The constructions and algorithms that we use extend their counterparts on standard \dfa s. In particular, complementation of saturated \fdfa s can be obtained on the same structure, and union and intersection is done on a product of the two given structures. The correctness proof of the latter is a bit subtle. 

As for the succinctness, which turns out to be more involved, we show that saturated \fdfa s properly lie in between deterministic and nondeterministic $\omega$-automata. We provide polynomial translations from deterministic $\omega$-automata to \fdfa s and from \fdfa s to nondeterministic $\omega$-automata, and show that an exponential state blowup in the opposite directions is inevitable in the worst case.

Specifically, a saturated \fdfa\ of size $(n,k)$ can always be transformed into an equivalent nondeterministic \buchi\ automaton (\nba) with $O(n^2 k^3)$ states. (Recall that an \fdfa\ of size $(n,k)$ can have up to $n+nk$ states in total, having $n$ states in the leading automaton and up to $k$ states in each of the $n$ progress automata.) As for the other direction, transforming an \nba\ with $n$ states to an equivalent \fdfa\ is shown to be in $2^{\Theta(n \log n)}$.
This is not surprising since, as shown by Michel~\cite{Michel88}, complementing an \nba\ involves a $2^{\Omega(n \log n)}$ state blowup, while \fdfa\ complementation requires no state blowup.

Considering deterministic $\omega$-automata, a \buchi\ or \cobuchi\ automaton (\dba\ or \dca) with $n$ states can be transformed into an equivalent \fdfa\ of size $(n,2n)$, and a deterministic parity automaton (\dpa) with $n$ states and $k$ colors can be transformed into an equivalent \fdfa\ of size $(n,kn)$. As for the other direction, since \dba\ and \dca\ do not recognize all the $\omega$-regular languages, while saturated \fdfa s do, a transformation from an \fdfa\ to a \dba\ or \dca\ need not exist. Comparing \fdfa s to \dpa s, which do recognize all $\omega$-regular languages, we get that \fdfa s can be exponentially more succinct:  We show a family of  languages $\{L_n\}_{n\geq 1}$, such that for every $n$, there exists an \fdfa\ of size $(n+1,n^2)$ for $L_n$, but any \dpa\ recognizing $L_n$ must have at least $2^{n-1}$ states. (A deterministic Rabin or Streett automaton for $L_n$ is also shown to be exponential in $n$, requiring at least $2^{\frac{n}{2}}$ states.)

\section{Preliminaries}
An \emph{alphabet} $\Sigma$ is a finite set of symbols. The set of finite words over $\Sigma$ is denoted by $\Sigma^*$, and the set of infinite words, termed $\omega$-words, over $\Sigma$ is denoted by $\Sigma^\omega$. As usual, we use $x^*$, $x^+$, and $x^\omega$ to denote finite, non-empty finite, and infinite concatenations of $x$, respectively, where $x$ can be a symbol or a finite word. We use $\epsilon$ for the empty word. An infinite word $w$ is \emph{ultimately periodic} if there are two finite words $u\in \Sigma^*$ and $v\in\Sigma^+$, such that $w=uv^\omega$. A \emph{language} is a set of finite words, that is, a subset of $\Sigma^*$, while an $\omega$-language is a set of $\omega$-words, that is, a subset of $\Sigma^\omega$. For natural numbers $i$ and $j$ and a word $w$, we use $[i..j]$ for the set $\{i, i+1, \ldots, j\}$, $w[i]$ for the $i$-th letter of $w$, and $w[i..j]$ for the subword of $w$ starting at the $i$-th letter and ending at the $j$-th letter, inclusive.

An \emph{automaton} is a tuple $\A=\la \Sigma, Q, \initstate, \delta \ra$ consisting of an alphabet $\Sigma$, a finite set $Q$ of
states, an initial state $\initstate\in Q$, and a transition function $\delta: Q \times \Sigma \rightarrow 2^Q$. 
A run of an automaton on a finite word $v=a_1 a_2\ldots a_n$ is a sequence of states ${q_0,q_1,\ldots,q_n}$ such that $q_0=\initstate$, and for each $i\geq 0$, ${q_{i+1}\in\delta(q_i,a_i)}$.  A run on an infinite word is defined similarly and results in an infinite sequence of states.
The transition function is naturally extended to a function $\delta: Q \times \Sigma^*\rightarrow 2^Q$, by defining $\delta(q,\epsilon)=\{q\}$, and $\delta(q,a v)=\cup_{p\in\delta(q,a)}\delta(p,v)$ for ${q\in Q}$, ${a\in\Sigma}$, and ${v\in\Sigma^*}$.
We often use $\A(v)$ as a shorthand for $\delta(\initstate,v)$ and $|\A|$ for the number of states in $Q$.
We use $\A^q$ to denote the automaton $\la \Sigma, Q, q, \delta \ra$ obtained from $\A$ by replacing the initial state with $q$. We say that $\A$ is \emph{deterministic} if $|\delta(q,a)|\leq1$ and \emph{complete} if $|\delta(q,a)|\geq1$, for every $q\in Q$ and $a\in\Sigma$. For simplicity, we consider all automata to be complete. (As is known, every automaton can be linearly translated to an equivalent complete automaton, with respect to the relevant equivalence notion, as defined below.)

By augmenting an automaton with an acceptance condition $\alpha$, thereby obtaining a tuple $\la \Sigma, Q, \initstate,$ $\delta, \alpha \ra$, we get an \emph{acceptor}, a machine that accepts some words and rejects others. An acceptor accepts a word if at least one of the runs on that word is accepting. For finite words the acceptance condition is a set $F \subseteq Q$ of \emph{accepting states}, and a run on a word $v$ is accepting if it ends in an accepting state, i.e., if $\delta(\initstate,v)$ contains an element of $F$. For infinite words, there are various acceptance conditions in the literature; here we mention three: \buchi, \cobuchi, and parity.\footnote{There are other acceptance conditions in the literature, the most known of which are weak, Rabin, Streett, and Muller. The three conditions that we concentrate on are the most used ones; \buchi\ and \cobuchi\ due to their simplicity, and parity due to being the simplest for which the deterministic variant is strong enough to express all $\omega$-regular languages.}
The \buchi\ and \cobuchi\ acceptance conditions are also a set $F \subseteq Q$. A run of a \buchi\ automaton is accepting if it visits $F$ infinitely often. A run of a co-\buchi\ automaton is accepting if it visits $F$ only finitely many times. A parity acceptance condition is a map $\kappa:Q \rightarrow [1..k]$ assigning each state a color (or rank). A run is accepting if the minimal color visited infinitely often is odd.
We use $\sema{\A}$ to denote the set of words accepted by a given acceptor $\A$, and say that $\A$ \emph{accepts} or \emph{recognizes} $\sema{\A}$. Two acceptors $\A$ and $\B$ are \emph{equivalent} if $\sema{\A}=\sema{\B}$.

We use three letter acronyms to describe acceptors, where the first letter is either $\textsc{d}$ or $\textsc{n}$ depending on whether the automaton is \emph{deterministic} or \emph{nondeterministic}, respectively. The second letter is one of $\{\textsc{f,b,c,p}\}$: $\textsc{f}$ if this is an acceptor over finite words, $\textsc{b}$, $\textsc{c}$, or $\textsc{p}$ if it is an acceptor over infinite words with \buchi, \cobuchi, or parity acceptance condition, respectively. The third letter is always $\textsc{a}$ for an acceptor (or automaton). 

For finite words, \nfa\ and \dfa\ have the same expressive power.  A language is said to be \emph{regular} if it is accepted by an \nfa. 
For infinite words, the theory is more involved. While \npa s, \dpa s, and \nba s have the same expressive power, \dba s, \nca s, and \dca s are strictly weaker than \nba s. 
An $\omega$-language is said to be \emph{$\omega$-regular} if it is accepted by an \nba.

\section{Families of DFAs (FDFAs)}
\label{sec:fdfas}
It is well known that two $\omega$-regular languages are equivalent if they agree on the set of ultimately periodic words (this is a consequence of McNaughton's theorem~\cite{McNaughton66}). An ultimately periodic word $uv^\omega$, where $u\in\Sigma^*$ and $v\in\Sigma^+$, is usually represented by the pair $(u,v)$. A canonical representation of an $\omega$-regular language can thus consider only ultimately periodic words, namely define a language of pairs $(u,v)\in \Sigma^* \times \Sigma^+$. Such a representation $\F$ should satisfy the \emph{saturation} property: considering two pairs $(u,v)$ and $(u',v')$, if $uv^\omega=u'v'^\omega$ then either both $(u,v)$ and $(u',v')$ are accepted by $\F$ or both are rejected by $\F$. 

A family of \dfa s (\fdfa) accepts such pairs $(u,v)$ of finite words. Intuitively, it consists of a deterministic \emph{leading automaton} $\Q$ with no acceptance condition that runs on the prefix-word $u$, and for each state $q$ of $\Q$, a \emph{progress automaton} $\P_q$, which is a \dfa\ that runs on the period-word $v$. 

A straightforward definition of acceptance for a pair $(u,v)$, could have been that the run of the leading automaton $\Q$ on $u$ ends at some state $q$, and the run of the progress automaton $\P_q$ on $v$ is accepting. This goes along the lines of $\ldollar$-automata \cite{CalbrixNP93}. However, such an acceptance definition does not fit well the saturation requirement, and might enforce very large automata \cite{AngluinF14}. The intuitive reason is that every progress automaton might need to handle the period-words of all prefix-words.

To better fit the saturation requirement, the acceptance condition of an \fdfa\ is defined with respect to a \emph{normalization} of the input pair $(u,v)$. 
The normalization is a new pair $(x,y)$, such that $xy^\omega=uv^\omega$, and in addition, the run of the leading automaton $\Q$ on $xy^i$ ends at the same state for every natural number $i$. Over the normalized pair $(x,y)$, the acceptance condition follows the straightforward approach discussed above. This normalization resembles the implicit flexibility in the acceptance conditions of $\omega$-automata, such as the \buchi\ condition, and allows saturated \fdfa s to be up to exponentially more succinct than $\ldollar$-automata \cite{AngluinF14}. 

Below, we formally define an \fdfa, the normalization of an input pair $(u,v)$, and the acceptance condition. We shall use $\Sigmasp$ as a shorthand for $\Sigma^* \times \Sigma^+$, whereby the input to an \fdfa\ is a pair $(u,v)\in\Sigmasp$.

\begin{defi}[A family of \dfa s (\fdfa)]\hspace{-0.5em}\footnote{The \fdfa s defined here follow the definition in~\cite{AngluinF14}, which is a little different from the definition of \fdfa s in \cite{Klarlund94}; the latter provide a machine model for the families of right congruences introduced in~\cite{MalerStaiger97}. The main differences between the two definitions are: i) In \cite{Klarlund94}, a pair $(u,v)$ is accepted by an \fdfa\ $\F=(\Q,\PP)$ if there is some factorization $(x,y)$ of $(u,v)$, such that $\Q(x)=q$ and $\P_q$ accepts $y$; and ii) in \cite{Klarlund94}, the \fdfa\ $\F$ should also satisfy the constraint that for all words $u\in\Sigma^*$ and $v,v'\in\Sigma^+$, if $\P_{\Q(u)}(v)= \P_{\Q(u)}(v')$ then $\Q(uv) = \Q(uv')$.}
\begin{itemize}
\item A \emph{family of \dfa s} (\fdfa) is a pair $(\Q, \PP)$, where $\Q=(\Sigma,Q,\initstate,\delta)$ is a deterministic \emph{leading} automaton, and $\PP$ is a set of $|\Q|$ \dfa s, including for each state $q\in Q$, a \emph{progress} \dfa\  $\P_q =(\Sigma,P_q,\initstate_q,\delta_q, F_q)$.
\item Given a pair $(u,v)\in\Sigmasp$ and an automaton $\A$, the \emph{normalization} of $(u,v)$ w.r.t $\A$ is the pair $(x,y)\in\Sigmasp$, such that $x = {u}{v}^i$, $y={v}^j$, and $i \geq 0$, $j\geq 1$ are the smallest numbers for which $\A({u}{v}^i)=\A({u}{v}^{i+j})$. (By ``smallest numbers'' $(i,j)$, one can consider lexicographic order, having the smallest $i$, and for it the smallest $j$. By Fine and Wilf's theorem, however, there is no ambiguity by just requiring    ``smallest numbers''. Notice that since we consider complete automata, such a unique pair $(x,y)$ is indeed guaranteed.)
\item Let $\F = (\Q,\PP)$ be an \fdfa, $(u,v)\in\Sigmasp$, and $(x,y)\in\Sigmasp$ the normalization of $(u,v)$ w.r.t\ $\Q$. We say that $(u,v)$ is \emph{accepted} by $\F$ iff $\Q(x)=q$ for some state $q$ of $\Q$ and $\P_{q}(y)$ is an accepting state of $\P_q$.
\item We use $\Lang{\F}$ to denote the set of pairs accepted by $\F$. 
\item We define the \emph{size} of $\F$, denoted by $|\F|$, as the pair $(|\Q|, \max \{|\P_q|\}_{q\in \Q})$.
\item An \fdfa\ $\F$ is \emph{saturated} if for every two pairs $(u,v)$ and $(u',v')$ such that $uv^\omega=u'v'^\omega$, either both $(u,v)$ and $(u',v')$ are in  $\Lang{\F}$ or both are not in $\Lang{\F}$.
\end{itemize}
\end{defi}
A saturated \fdfa\ can be used to characterize an $\omega$-regular language (see Theorem~\ref{thm:Characterize}), while an unsaturated \fdfa\ cannot.

An unsaturated \fdfa\ is depicted in Figure~\ref{fig:UnsaturatedFdfa} on the left. Consider the pairs $(b,a)$ and $(ba,aa)$.
The former is normalized into $(b,aa)$, as the run of the leading automaton $\U$ on ``$b$'' reaches the state $l$, and then when $\U$ iterates on ``$a$'', the first state to be revisited is $l$, which happens after two iterations. The latter is normalized into $(ba,aa)$, namely it was already normalized. Now, $(b,a)$ is accepted since in the run on its normalization $(b,aa)$,  $\U$ reaches the state $l$ when running on ``$b$'', and the progress automaton $\P_l^\U$ accepts ``$aa$''. On the other hand, $(ba,aa)$ is not accepted since the run of $\U$ on ``$ba$'' reaches the state $r$, and the progress automaton $\P_r^\U$ does not accept ``$aa$''. Yet, $ba^\omega=ba(aa)^\omega$, so the \fdfa\ should have either accepted or rejected both $(b,a)$ and $(ba,aa)$, were it saturated, which is not the case.

A saturated \fdfa\ is depicted in Figure~\ref{fig:SaturatedFdfa} on the right. It accepts pairs of the forms $(\Sigma^*, a^+)$ and $(\Sigma^*, b^+)$, and characterizes the $\omega$-regular language $(a+b)^*(a^{\omega} + b^{\omega})$ of words in which there are eventually only $a$'s or only $b$'s.

\begin{figure}[t]
\centering
\scalebox{0.7}{
\begin{tikzpicture}[->,>=stealth',shorten >=1pt,auto,node distance=1.8cm,semithick,initial text=,initial where=left]

% An un saturated FDFA

% Leading
\node[label]           			(_Q)    					{};
\node[label]           			(Q)   [below of=_Q]	{$\U:$};
\node[state,initial]           	(Q1) [right of=Q]  		{$l$};
\node[state]           			(Q2) [right of=Q1]  		{$r$};

\path (Q1) [loop above] 	edge node {$b$} 	(Q1);
\path (Q1) [bend left]		edge node {$a$} 	(Q2);
\path (Q2) [bend left] 		edge node {$a,b$} (Q1);

% P_l
\node[label]           				(_Pl) [right of=Q2, node distance=-1cm]   	{};
\node[label]           				(__Pl) [above of=_Pl]  {};
\node[label]           				(Pl)   [right of=__Pl] 	{$\P_l^\U:$};
\node[state,initial,accepting]	(Pl1) [right of=Pl]  		{};

\path (Pl1)[loop below] 	edge node {$a,b$} (Pl1);

%P_r
\node[label]           				(_Pr) [below of=Pl]  	{};
\node[label]           				(Pr)   [below of=_Pr]  	{$\P_r^\U:$};
\node[state,initial]           		(Pr1) [right of=Pr]  		{};
\node[state,accepting]          (Pr2) [right of=Pr1]  	{};

\path (Pr1) [loop above] 	edge node {$a$} 	(Pr1);
\path (Pr1) 					edge node {$b$} 	(Pr2);
\path (Pr2) [loop above] 	edge node {$a,b$} (Pr2);

% A saturated FDFA

% Leading
\node[label]           			(_SQ) [right of=_Q, node distance=11cm]   					{};
\node[label]           			(SQ)   [below of=_SQ]	{$\S:$};
\node[state,initial]           	(SQ1) [right of=SQ]  		{$l$};
\node[state]           			(SQ2) [right of=SQ1]  		{$r$};

\path (SQ1) [loop above] 	edge node {$a$} 	(SQ1);
\path (SQ1) [bend left]		edge node {$b$} 	(SQ2);
\path (SQ2) [loop above] 	edge node {$b$} 	(SQ2);
\path (SQ2) [bend left] 		edge node {$a$}	(SQ1);

% P_l
\node[label]           				(_SPl) [right of=SQ2, node distance=-1cm]   	{};
\node[label]           				(__SPl) [above of=_SPl]  {};
\node[label]           				(SPl)   [right of=__SPl] 	{$\P_l^\S:$};
\node[state,initial,accepting]	(SPl1) [right of=SPl]  		{};
\node[state]          				(SPl2) [right of=SPl1]  	{};

\path (SPl1) [loop below] 	edge node {$a$} 	(SPl1);
\path (SPl1) 					edge node {$b$} 	(SPl2);
\path (SPl2) [loop below] 	edge node {$a,b$} (SPl2);

%P_r
\node[label]           				(_SPr) [below of=SPl]  	{};
\node[label]           				(SPr)   [below of=_SPr]  	{$\P_r^\S:$};
\node[state,initial,accepting]	(SPr1) [right of=SPr]  		{};
\node[state]          				(SPr2) [right of=SPr1]  	{};

\path (SPr1) [loop above] 	edge node {$b$} 	(SPr1);
\path (SPr1) 					edge node {$a$} 	(SPr2);
\path (SPr2) [loop above] 	edge node {$a,b$} (SPr2);

\end{tikzpicture}
}
\caption{Left: an unsaturated \fdfa\ with the leading automaton $\U$ and progress \dfa s $\P_l^\U$ and $\P_r^\U$. Right: a saturated \fdfa\ with the leading automaton $\S$ and progress \dfa s $\P_l^\S$ and $\P_r^\S$. 
}
\label{fig:UnsaturatedFdfa}
\label{fig:SaturatedFdfa}
\end{figure}
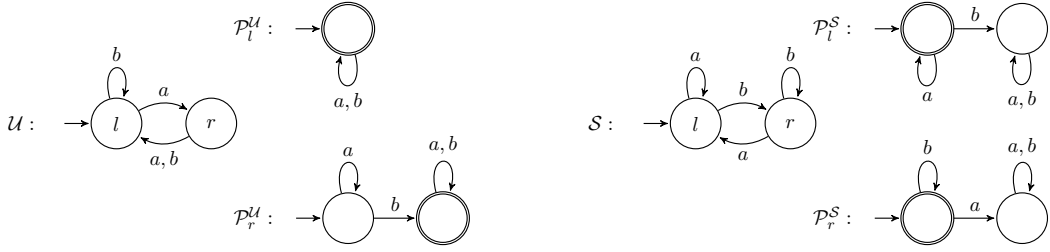

\section{Boolean Operations and Decision Procedures}
We provide below algorithms for performing the Boolean operations of complementation, union, and intersection on saturated \fdfa s, and deciding the basic questions on them, such as emptiness, universality, and language containment. All of these algorithms can be done in nondeterministic logarithmic space, taking advantage of the partial deterministic nature of \fdfa s.\footnote{Another model that lies in between deterministic and nondeterministic automata are ``semi-deterministic B\"uchi automata'' \cite{VardiWolper86}, which are B\"uchi automata that are deterministic in the limit: from every accepting state onward, their behaviour is deterministic. Yet, as opposed to \fdfa s, complementation of semi-deterministic B\"uchi automata might involve an exponential state blowup \cite{BHSST16}.} We conclude the section with the decision problem of whether an arbitrary \fdfa\ is saturated, showing that it can be resolved in polynomial space.

\subsubsection*{Boolean operations}
Saturated \fdfa s are closed under Boolean operations as a consequence of Theorem~\ref{thm:Characterize}, which shows that they characterize exactly the set of $\omega$-regular languages. We provide below explicit algorithms for these operations, showing that they can be done effectively.

Complementation of an \fdfa\ is simply done by switching between accepting and non-accepting states in the progress automata, as is done with \dfa s.

\begin{thm}\label{thm:fdfa-complementation}
Let $\F$ be an \fdfa. There is a constant-space algorithm to obtain an \fdfa\ $\F^c$, such that $\Lang{\F^c}=\Sigmasp\setminus\Lang{\F}$, $|\F^c|=|\F|$, and $\F^c$ is saturated iff $\F$ is.
\end{thm}
\begin{proof}
Let $\F=(\Q,\PP)$, where for each state $q$ of $\Q$, $\PP$ has the \dfa\ $\P_q=(\Sigma,P_q,\initstate_q,\delta_q,F_q)$. We define $\F^c$ to be the \fdfa\ $(\Q,\PP^c)$, where for each state $q$ of $\Q$, $\PP^c$ has the \dfa\ $\P^c_q=(\Sigma,P_q,\initstate_q,\delta_q, P_q\setminus F_q)$.
We claim that $\F^c$ recognizes the complement language of $\F$. Indeed, let $(u,v)\in\Sigmasp$ and $(x,y)$ its normalization with respect to $\Q$. Then $(u,v)\in\Lang{\F}$ iff $y\in\Lang{\P_{\Q(x)}}$. Thus $(u,v)\notin\Lang{\F}$ iff $y\notin\Lang{\P_{\Q(x)}}$ iff $y\in\Lang{\P^c_{\Q(x)}}$ iff $(u,v)\in\Lang{\F^c}$. 

Since $\F$ is saturated, so is $\F^c$, as for all pairs $(u,v)$ and $(u',v')$ such that $uv^\omega=u'v'^\omega$, $\F$ either accepts or rejects them both, implying that $\F^c$ either rejects or accepts them both, respectively.
\end{proof}

Union and intersection of saturated \fdfa s also resemble the case of \dfa s, and are done by taking the product of the leading automata and each pair of progress automata. Yet, the correctness proof is a bit subtle, and relies on the following lemma, which shows that for a normalized pair $(x,y)$, the period-word $y$ can be manipulated in a certain way, while retaining normalization.

\begin{lem}
\label{lem:normalization}
Let $\Q$ be an automaton, and let $(x,y)$ be the normalization of some $(u,v)\in\Sigmasp$ w.r.t.\ $\Q$.
Then for every $i\geq 0$, $j\geq 1$ and finite words $y',y''$ such that $y=y'y''$, we have that $(xy^iy',(y''y')^j)$ is the normalization of itself w.r.t.\ $\Q$.
\end{lem}

\begin{proof}
Let $x_1=xy^iy'$ and $y_1=(y''y')^j$. 
Since $(x,y)$ is normalized w.r.t.\ $\Q$, we know that the run of $\Q$ on $xy^\omega$ is of the form $q_0,q_1,\ldots,q_{k-1}, (q_k, q_{k+1},$ $q_{k+2}, \ldots, q_m)^\omega$, where $|x|=k$ and $|y|=(m-k)+1$. 
As $xy^\omega=x_1y_1^\omega$, the run of $\Q$ on $x_1 y_1^\omega$ is identical. Since $x$ is a prefix of $x_1$, the position $|x_1|$ lies within the repeated period, implying that $|x_1|$ is the first position, from $|x_1|$ onwards, that is repeated along the aforementioned run. Since $y_1$ is a cyclic repetition of $y$, and $\Q$ loops back over $y$, it also loops back over $y_1$. Thus $(x_1,y_1)$ is  the normalization of itself w.r.t.\ $Q$.
\end{proof}

We continue with the union and intersection of saturated \fdfa s.

\begin{thm}\label{thm:fdfa-union-and-intersection}
Let $\F_1$ and $\F_2$ be saturated \fdfa s of size $(n_1,k_1)$ and $(n_2,k_2)$, respectively. There exist logarithmic-space algorithms to obtain saturated \fdfa s $\H$ and $\H'$ of size $(n_1 n_2,\ k_1  k_2)$, such that $\Lang{\H}=\Lang{\F_1}\cap\Lang{\F_2}$ and $\Lang{\H'}=\Lang{\F_1}\cup\Lang{\F_2}$.
\end{thm}
\begin{proof}
The constructions of the union and intersection of \fdfa s are similar, only differing by the accepting states. We shall thus describe them together.

\subproof{Construction}

Given two automata $\A_1$ and $\A_2$, where $\A_i = (\Sigma, A_i,\initstate_i,\delta_i)$, we denote by $\A_1\Times\A_2$ the product automaton
$(\Sigma,A_1\Times A_2,(\initstate_1,\initstate_2),\delta_\times)$, where for every $\sigma\in\Sigma$, $\delta_\times((q_1,q_2),\sigma)=(\delta(q_1,\sigma),\delta(q_2,\sigma))$. 

Given two \dfa s $\D_1=(\A_1,F_1)$ and $\D_2=(\A_2,F_2)$, over the automata $\A_1$ and $\A_2$, and with the accepting states $F_1$ and $F_2$, respectively, we define the \dfa s $\D_1\otimes \D_2$ and $\D_1\oplus \D_2$ as follows:
\begin{itemize}
\item $\D_1\otimes \D_2 = (\A_1\Times\A_2,F_1 \Times F_2)$
\item $\D_1\oplus \D_2 = (\A_1\Times\A_2,F_1 \Times A_2 \cup A_1 \Times F_2)$
\end{itemize}
Given two sets of \dfa s $\PP_1$ and $\PP_2$, we define the sets of \dfa s $\PP_1\otimes \PP_2$ and $\PP_1\oplus \PP_2$ as follows:
\begin{itemize}
\item $\PP_1\otimes \PP_2 = \{ \D_1\otimes \D_2 \ST \D_1 \in \PP_1 \mbox{ and } \D_2 \in \PP_2\}$
\item $\PP_1\oplus \PP_2 = \{ \D_1\oplus \D_2 \ST \D_1 \in \PP_1 \mbox{ and } \D_2 \in \PP_2\}$
\end{itemize}
Given saturated \fdfa s  $\F_1=(\Q_1,\PP_1)$ and $\F_2=(\Q_2,\PP_2)$, we claim that $\H=(\Q_1 \times \Q_2, \PP_1 \otimes \PP_2)$ and $\H'=(\Q_1 \times \Q_2, \PP_1 \oplus \PP_2)$ are saturated \fdfa s that recognize the intersection and union of $\Lang{\F_1}$ and $\Lang{\F_2}$, respectively. 

Notice that the number of states in $\H$ and $\H'$ is quadratic in the number of states in $\F_1$ and $\F_2$, yet the algorithm can generate the representation of $\H$ and $\H'$ in space logarithmic in the size of $\F_1$ and $\F_2$: It sequentially traverses the states of $\Q_1$, and for each of its states, it sequentially traverses the states of $\Q_2$. Thus, it should only store the currently traversed states in $\Q_1$ and $\Q_2$, where each of them requires a storage space of size logarithmic in the number of states in $\Q_1$ and $\Q_2$, respectively. The same holds for generating the product of $\PP_1$ and $\PP_2$.

\subproof{Correctness}

Consider a pair $(u,v)\in\Sigmasp$. Let  $(x_1,y_1)$ and $(x_2,y_2)$ be its normalization with respect to $\Q_1$ and $\Q_2$, respectively, where $x_1={u}{v}^{i_1}$, $y_1={v}^{j_1}$, $x_2={u}{v}^{i_2}$, and $y_2={v}^{j_2}$.
Let $i = \max(i_1,i_2)$ and $j$ be the least common multiple of $(j_1,j_2)$. Define $x={u}{v}^{i}$ and $y={v}^{j}$. 

Observe that the normalization of $(u,v)$ with respect to $\Q_1 \Times \Q_2$ is $(x,y)$: i) The equality $\Q_1 \Times \Q_2(x)=\Q_1 \Times \Q_2(xy)$ follows from taking $i$ to be bigger than both $i_1$ and $i_2$, which guarantees that further concatenations of $v$ will be along a cycle w.r.t.\ both $Q_1$ and $Q_2$, and taking $j$ to be multiple of both $j_1$ and $j_2$, which guarantees that both $Q_1$ and $Q_2$ complete a cycle along $y$. ii) The minimality of $i$ follows from the fact that it is equal to either $i_1$ or $i_2$, as a smaller number will contradict the minimality of either $i_1$ or $i_2$, and the minimality of $j$ follows from the fact that it is the minimal number divided by both $j_1$ and $j_2$, as a number not divided by one of them will not allow either $Q_1$ or $Q_2$ to complete a cycle. 

We have ${\Q_1\Times\Q_2}\,(xy)={\Q_1\Times\Q_2}\,(x)=(\Q_1(x),\Q_2(x))$. 
Since $xy^\omega=x_1y_1^\omega$ and $\F_1$ is saturated, we get that $(x,y)\in\Lang{\F_1}$ iff $(x_1,y_1)\in\Lang{\F_1}$. 
Since the pair $(x,y)$ satisfies the requirements of Lemma~\ref{lem:normalization} w.r.t.\ $(x_1,y_1)$ and $\Q_1$, it follows that $(x,y)$ is a normalization of itself w.r.t.\ $\Q_1$. Thus, $y\in\P_{\Q_1(x)}$ iff $y_1\in\P_{\Q_1(x_1)}$.
Analogously, $y\in\P_{\Q_2(x)}$ iff $y_2\in\P_{\Q_2(x_2)}$.  

Hence, $(u,v)\in\Lang{\F_1}\cap \Lang{F_2}$ iff 
 $(y_1\in\P_{\Q_1(x_1)}$ and $y_2\in\P_{\Q_2(x_2)})$ iff
 $y\in \P_{\Q_1(x)}\otimes \P_{\Q_2(x)}$ iff
$(u,v)\in\H$. 
Similarly, 
$(u,v)\in\Lang{\F_1}\cup \Lang{F_2}$ iff 
 $(y_1\in\P_{\Q_1(x_1)}$ or $y_2\in\P_{\Q_2(x_2)})$ iff
$y\in \P_{\Q_1(x)}\oplus \P_{\Q_2(x)}$  iff $(u,v)\in\H'$.  

The saturation of $\H$ and $\H'$ directly follows from the above proof of the languages they recognize: consider two pairs $(u,v)$ and $(u',v')$, such that $uv^\omega=u'v'^\omega$. Then, by the saturation of $\F_1$ and $\F_2$, both pairs either belong, or not, to each of $\Lang{\F_1}$ and $\Lang{\F_2}$. Hence, both pairs belong, or not, to each of $\Lang{\H}=\Lang{\F_1}\cap\Lang{\F_2}$ and $\Lang{\H'}=\Lang{\F_1}\cup\Lang{\F_2}$.
\end{proof}

\subsubsection*{Decision procedures}
All of the basic decision problems can be resolved in nondeterministic logarithmic space, using the Boolean operations above and corresponding decision algorithms for \dfa s.

The first decision question to consider is that of \emph{membership}: given a pair $(u,v)$ and an \fdfa\ $\F=(\Q,\PP)$, does $\F$ accept $(u,v)$? The question is answered by normalizing $(u,v)$ into a pair $(x,y)$ and evaluating the runs of $\Q$ over $x$ and of $\P_{\Q(x)}$ over $y$. A normalized pair is determined by traversing along $\Q$, making up to $|Q|$ repetitions of $v$. Notice that memory wise, $x$ and $y$ only require a logarithmic amount of space, as they are of the form $x=uv^i$ and $y=v^j$, where the representation of $i$ and $j$ is bounded by $\log |Q|$.
The overall logarithmic-space solution follows from the complexity of algorithms for deterministically traversing along an automaton. 

\begin{prop}\label{prop:Membership}
Given a pair $(u,v)\in\Sigmasp$ and an \fdfa\ $\F$ of size $(n,k)$, the membership question, of whether $(u,v)\in\Lang{\F}$, can be resolved in deterministic space of $O(\log n+\log k)$.
\end{prop}

The next questions to consider are those of \emph{emptiness} and \emph{universality}, namely given an \fdfa\ $\F=(\Q, \PP)$, whether $\Lang{\F} = \emptyset$, and whether $\Lang{\F} = \Sigmasp$, respectively. Notice that the universality problem is equivalent to the emptiness problem over the complement of $\F$. For nondeterministic automata, the complement automaton might be exponentially larger than the original one, making the universality problem much harder than the emptiness problem. Luckily, \fdfa\ complementation is done in constant space, as is the case with deterministic automata, making the emptiness and universality problems equally easy.

The emptiness problem for an \fdfa\ $(\Q,\PP)$ cannot be resolved by only checking whether there is a nonempty progress automaton in $\PP$, since it might be that the accepted period $v$ is not part of any normalized pair. Yet, the existence of a prefix-word $x$ and a period-word $y$, such that $\Q(x)=\Q(xy)$ and $\P_{\Q(x)}$ accepts $y$ is a sufficient and necessary criterion for the nonemptiness of $\F$.
This can be tested in NLOGSPACE. Hardness in NLOGSPACE follows by a reduction from graph reachability \cite{Jon75}.

\begin{thm}\label{thm:non-emptiness}
Emptiness and universality for \fdfa s are NLOGSPACE-complete.
\end{thm}
\begin{proof}
An \fdfa\ $\F=(\Q,\PP)$ is not empty iff there exists a pair $(u,v)\in\Sigmasp$, whose normalization is some pair $(x,y)$, such that $\P_{\Q(x)}$ accepts $y$. By Lemma~\ref{lem:normalization}, a normalized pair is a normalization of itself, implying that a sufficient and necessary criterion for the nonemptiness of $\F$ is the existence of a pair $(x,y)$, such that  $\Q(x)=\Q(xy)$ and $\P_{\Q(x)}$ accepts $y$.

We can nondeterministically find such a pair $(x,y)$ in logarithmic space by guessing $x$ and $y$ (a single letter at each step), and traversing along $\Q$ and $\P_{\Q(x)}$ \cite{GJ79}.

Hardness in NLOGSPACE follows by a reduction from graph reachability \cite{Jon75}, taking an \fdfa\ whose leading automaton has a single state.

As \fdfa\ complementation is done in constant space (Theorem~\ref{thm:fdfa-complementation}), the universality problem has the same space complexity.
\end{proof}

The last decision questions we handle are those of \emph{equality} and \emph{containment}, namely given saturated \fdfa s $\F$ and $\F'$, whether $\Lang{\F}=\Lang{\F'}$ and whether $\Lang{\F}\subseteq\Lang{\F'}$, respectively. Equality reduces to containment, as $\Lang{\F}=\Lang{\F'}$  iff $\Lang{\F}\subseteq\Lang{\F'}$ and $\Lang{\F'}\subseteq\Lang{\F}$. Containment can be resolved by intersection, complementation, and emptiness check, as $\Lang{\F}\subseteq\Lang{\F'}$ iff $\Lang{\F} \cap \Lang{\F'}^c = \emptyset$. Hence, by Theorems \ref{thm:fdfa-complementation}, \ref{thm:fdfa-union-and-intersection}, and \ref{thm:non-emptiness}, these problems are NLOGSPACE-complete. Note that NLOGSPACE hardness immediately follows by reduction from the emptiness problem, which asks whether $\Lang{\F}=\emptyset$.
The complexity for \emph{equality} and \emph{containment} is easily derived from that of emptiness, intersection and complementation.

\begin{prop}\label{prop:containment}
Equality and containment for saturated \fdfa s are NLOGSPACE-complete.
\end{prop}

\subsubsection*{Saturation check}
All of the operations and decision problems above assumed that the given \fdfa s are saturated. This is indeed the case when learning \fdfa s via the algorithm of \cite{AngluinF14}, and when translating $\omega$-automata to \fdfa s (see Section~\ref{sec:Translations}). We show below that the decision problem of whether an arbitrary \fdfa\ is saturated is in PSPACE. We leave the question of whether it is PSPACE-complete open.

\begin{thm}\label{thm:SaturationIsInPspace}
The problem of deciding whether a given \fdfa\ is saturated is in PSPACE.
\end{thm}

\begin{proof}
Let $\F=(\Q,\PP)$ be an \fdfa\ of size $(n,k)$. We first show that if $\F$ is unsaturated then there exist words $u,v',v''$ such that $|u|\leq n$ and $|v'|,|v''|\leq n^nk^{2k}$, and non-negative integers $l,r\leq k$ such that $( u, (v'v'')^l )\in \Lang{\F}$ while $(uv',(v''v')^r)\notin \Lang{\F}$.	

If $\F$ is unsaturated then there exists some ultimately periodic word $w\in\Sigma^\omega$ that has two different decompositions to prefix and periodic words on which $\F$ provides different answers.
Let $\P$ and $\P'$ be the respective progress automata, corresponding to some states $q$ and $q'$ of $\Q$.
Let the run of $\Q$ on $w$ be $q_0,q_1,q_2,\ldots$. Since $w$ is ultimately periodic, there exist $i,j\in \NN$ such that $q_{h+j}=q_h$ for all $h>i$. That is, eventually the run cycles through a certain sequence of states. Then $q$ and $q'$ must be on the cycle where $w$ settles. Let $v'$ and $v''$ be the subwords of $w$ that are read on the part of the shortest such cycle from $q$ to $q'$ and from $q'$ back to $q$, respectively. (In case that $q=q'$, we let $v'=\epsilon$ and $v''$ be the whole cycle.) Then the different decompositions are of the form $(u,(v'v'')^l)$ and $(uv',(v''v')^r)$ where $u$ is a string that takes $Q$ to $q$. Let $l$ and $r$ be the shortest such, then since $\P$ and $\P'$ have at most $k$ states, we can assume $l,r\leq k$. We can also assume $u$ is a shortest such string and thus $|u|\leq n$. 

For a \dfa\ $\A=\la \Sigma, Q, \initstate, \delta, F\ra$ and a word $v\in\Sigma^*$, we use $\chi^{\A}_v$ to denote the function from $Q$ to $Q$ defined as $\chi^{\A}_v(q)=\delta(q,v)$. Note that given $|Q|=n$ there are at most $n^n$ different such functions.
Let $\X=\{\chi_v \ST  \chi_v = \la \chi^\Q_v, \chi^{\P}_v, \chi^{\P'}_v \ra, \ v\in\Sigma^*\}$. Then $\X$ is the set of congruence classes of the relation $v_1 \congr{\X} v_2$ iff $\chi^{\Q}_{v_1} = \chi^{\Q}_{v_2}$, $\chi^{\P}_{v_1} = \chi^{\P}_{v_2}$, and $\chi^{\P'}_{v_1} = \chi^{\P'}_{v_2}$.  We can build an automaton such that each state corresponds to a class in $\X$, the initial state is $\chi_\epsilon$, and the transition relation is $\delta_\X(\chi_w,\sigma)=\chi_{w\sigma}$. The cardinality of $\X$ is at most $n^n{k}^{2k}$. Thus, every state has a representative word of length at most $n^n{k}^{2k}$ taking the initial state to that state. Therefore, there exist words $y',y''$ such that $y'\congr{\X} v'$ and $y''\congr{\X} v''$ and $|y'|,|y''| \leq n^n{k}^{2k}$. Thus $u,y',y''$ and $l,r$ satisfy the promised bounds.

Now, to see that \fdfa\ saturation is in PSPACE note we can construct an algorithm that guesses integers $l,r \leq k$ and words $u,v',v''$ such that $|u| \leq n$ and $|v'|,|v''| \leq n^nk^{2k}$. It guesses the words letter by letter and constructs on the way  $\chi_{v'v''}$  and $\chi_{v''v'}$. It also constructs along the way the states $q'$ and $q$ such that $q=\delta(u)$ and $q'=\delta(uv')$. It then computes the $l$ and $r$ powers of $\chi_{v'v''}$ and $\chi_{v'',v'}$, respectively.
	Finally, it checks whether one of $(\chi^\P_{v'v''})^l(q)$ and $(\chi^{\P'}_{v''v'})^r(q')$ is accepting and the other is not, and if so returns that $\F$ is unsaturated.
	The required space is $O(nk^2 \log nk^2)$. This shows that saturation is in coNPSPACE, and by Savitch's and Immerman-Szelepcs\'{e}nyi's theorems, in PSPACE.
\end{proof}

\section{Translating To and From $\omega$-Automata}\label{sec:Translations}
As two $\omega$-regular languages are equivalent iff they agree on the set of ultimately periodic words \cite{McNaughton66}, an $\omega$-regular language can be characterized by a language of pairs of finite words, and in particular by a saturated \fdfa. We shall write $L\Characterizes L'$ to denote that a language $L\subseteq\Sigmasp$ characterizes an $\omega$-regular language $L'$. Formally:

\begin{defi}
A language $L\subseteq\Sigmasp$ \emph{characterizes} an $\omega$-regular language $L'\subseteq\Sigma^\omega$, denoted by $L\Characterizes L'$, if for every pair $(u,v)\in L$, we have $uv^\omega\in L'$, and for every ultimately periodic word $uv^\omega\in L'$, we have $(u,v)\in L$.
\end{defi}

The families of \dfa s defined in \cite{Klarlund94}, as well as the analogous families of right congruences of \cite{MalerStaiger97}, are known to characterize exactly the set of $\omega$-regular languages \cite{Klarlund94,MalerStaiger97}. This is also the case with our definition of saturated \fdfa s.
\begin{thm}\label{thm:Characterize}
Every saturated \fdfa\ characterizes an $\omega$-regular language, and for every $\omega$-regular language, there is a saturated \fdfa\ characterizing it.
\end{thm}
\begin{proof}
The two directions are proved in Theorems \ref{thm:dpa-to-fdfa} and \ref{thm:fdfa-to-nba}, below.
\end{proof}

In this section, we analyze the state blowup involved in translating deterministic and nondeterministic $\omega$-automata into equivalent saturated \fdfa s, and vice versa. For nondeterministic automata, we consider the \buchi\ acceptance condition, since it is the simplest and most commonly used among all acceptance conditions. For deterministic automata, we consider the parity acceptance condition since it is the simplest among all acceptance conditions whose deterministic version is equi-expressible to the $\omega$-regular languages. We also consider deterministic \buchi\ and \cobuchi, for the simple sub-classes they recognize.

\subsection{From $\omega$-Automata to FDFAs}
We show that \dba, \dca, and \dpa\ have polynomial translations to saturated \fdfa s, whereas translation of \nba s to \fdfa s may involve an inevitable exponential blowup.

\subsubsection*{From deterministic $\omega$-automata}
The constructions of a saturated \fdfa\ that characterizes a given \dba, \dca, or \dpa\ $\D$ share the same idea: The leading automaton is equivalent to $\D$, except for ignoring the acceptance condition, and each progress automaton consists of several copies of $\D$, memorizing the acceptance level of the period-word. For a \dba\ or a \dca, two such copies are enough, memorizing whether or not a \buchi\ (\cobuchi) accepting state was visited. For a \dpa\ with $k$ colors, $k$ such copies are required.

We start with the constructions of an \fdfa\ for a given \dba\ or \dca, which are almost the same. (\buchi\ and \cobuchi\ automata are special cases of parity automata, and therefore their translations to an \fdfa, as described in Theorem~\ref{thm:dba-to-fdfa}, are special cases of the translation given in Theorem~\ref{thm:dba-to-fdfa}. Nevertheless, for clarity reasons, we do give their explicit translations below.)

\begin{thm}\label{thm:dba-to-fdfa}
Let $\D$ be a \dba\ or a \dca\ with $n$ states. There exists a saturated \fdfa\ $\F$ of size $(n,2n)$, such that $\Lang{\F}\Characterizes\Lang{\D}$.
\end{thm}
\begin{proof}\

\subproof{Construction}
Let $\D=\la \Sigma, Q, \initstate, \delta, \alpha \ra$ be a \dba\ or a \dca. We define the \fdfa\ $\F=(\Q, \PP)$, where $\Q$ is the same as $\D$ (without acceptance), and each progress automaton $\P_q$ has two copies of $\D$, having $(q,0)$ as its initial state, and moving from the first to the second copy upon visiting a $\D$-accepting state. Formally: 
$\Q=\la \Sigma, Q, \initstate, \delta \ra$, and for each state $q\in Q$, $\PP$ has the \dfa\ $\P_q = \la \Sigma, Q\Times\{0,1\}, (q,0), \delta', F \ra$, where for every $\sigma\in\Sigma, \delta'((q,0),\sigma) = (\delta(q,\sigma),0)$ if $\delta(q,\sigma)\not\in\alpha$ and $(\delta(q,\sigma),1)$ otherwise; and $\delta'((q,1),\sigma) = (\delta(q,\sigma),1)$. The set $F$ of accepting states is $Q\Times\{1\}$ if $\D$ is a \dba\ and $Q\Times\{0\}$ if $\D$ is a \dca.

\subproof{Correctness}
We show the correctness for the case that $\D$ is a \dba. The case that $\D$ is a \dca\ is analogous.

Consider a word $uv^\omega\in\Lang{\D}$, and let $(x,y)$ be the normalization of $(u,v)$ w.r.t.\ $\Q$. Since $xy^\omega=uv^\omega\in\Lang{\D}$, it follows that $\D$ visits an accepting state when running on $y$ from the state $\D(x)$, implying that $\P_{\Q(x)}(y)$ is an accepting state. Hence, $(u,v)\in\Lang{\F}$.

As for the other direction, consider a pair $(u,v)\in\Lang{\F}$, and let $(x,y)$ be the normalization of $(u,v)$ w.r.t.\ $\Q$. Since $\P_{\Q(x)}(y)$ is an accepting state, $\D$ has the same structure as $\Q$, and $\Q(x)=\Q(xy)$, it follows that $\D$ visits an accepting state when running on $y$ from the state $\D(x)$, implying that $xy^\omega=uv^\omega\in\Lang{\D}$.

Note that $\F$ is saturated as a direct consequence of the proof that it characterizes an $\omega$-regular language.
\end{proof}

We continue with the construction of an \fdfa\ for a given \dpa.

\begin{thm}\label{thm:dpa-to-fdfa}
Let $\D$ be a \dpa\ with $n$ states and $k$ colors. There exists a saturated \fdfa\ $\F$ of size $(n,kn)$, such that $\Lang{\F}\Characterizes\Lang{\D}$.
\end{thm}
\begin{proof}\

\subproof{Construction}
Let $\D=\la \Sigma, Q, \initstate, \delta, \kappa \ra$ be a \dpa, where $\kappa:Q \rightarrow [1..k]$. We define the \fdfa\ $\F=(\Q, \PP)$, where $\Q$ is the same as $\D$ (without acceptance), and each progress automaton $\P_q$ has $k$ copies of $\D$, having $(q,\kappa(q))$ as its initial state, and moving to  a $j$-th copy upon visiting a state with color $j$, provided that $j$ is lower than the index of the current copy. The accepting states are those of the odd copies.

Formally: $\Q=\la \Sigma, Q, \initstate, \delta \ra$, and for each state $q\in
Q$, $\PP$ has the \dfa\ $\P_q = \la \Sigma, Q~\Times~[1..k],$ $ (q,\kappa(q)), \delta', F \ra$, where for every $\sigma\in\Sigma$ and $i\in[1..k]$, \[\delta'((q,i),\sigma) = (\delta(q,\sigma),\min(i, \kappa(\delta(q,\sigma)))).\]
The set $F$ of accepting states is $\{Q\Times\{i\} \ST i \mbox{ is odd } \}$.

\subproof{Correctness}
Analogous to the arguments in the proof of Theorem~\ref{thm:dba-to-fdfa}.
\end{proof}

\subsubsection*{From nondeterministic $\omega$-automata}
An \nba\ $\A$ can be translated into a saturated \fdfa\ $\F$, by first determinizing $\A$ into an equivalent \dpa\ $\A'$ \cite{Piterman06,FismanL15} (which might involve a $2^{O(n\log n)}$ state blowup and $O(n)$ colors \cite{SCH09}), and then polynomially translating $\A'$ into an equivalent \fdfa\ (Theorem~\ref{thm:dpa-to-fdfa}).

\begin{prop}\label{prop:buchi-to-recurrent}
Let $\B$ be an \nba\ with $n$ states. There is a saturated \fdfa\ that characterizes $\Lang{\B}$ with a leading automaton and progress automata of at most $2^{O(n\log n)}$ states each.
\end{prop}

A $2^{O(n\log n)}$ state blowup in this case is inevitable, based on the lower bound for complementing \nba s \cite{Michel88,Yan06,Sch09b}, the constant complementation of \fdfa s, and the polynomial translation of a saturated \fdfa\ to an \nba:

\begin{thm}
There exists a family of \nba s $\B_1, \B_2, \ldots$, such that for every $n\in\NN$, $\B_n$ is of size $n$, while a saturated \fdfa\ that characterizes $\Lang{\B_n}$ must be of size $(m,k)$, such that $\max(m,k) \geq 2^{\Omega(n \log n)}$.
\end{thm}

\begin{proof}
Michel~\cite{Michel88} has shown that there exists a family of languages $\{L_n\}_{n\geq 1}$, such that for every $n$, there exists an \nba\ of size $n$ for $L_n$, but an \nba\ for $L_n^c$, the complement of $L_n$, must have at least $2^{n \log n}$ states. 

Assume, towards a contradiction, that exist $n\in\NN$ and a saturated \fdfa\ $\F$ of size $(m,k)$ that characterizes $L_n$, such that $\max(m,k) < 2^{\Omega(n\log n)}$. Then, by Theorem~\ref{thm:fdfa-complementation}, there is a saturated \fdfa\ $\F^c$ of size $(m,k)$ that characterizes $L_n^c$. Thus, by Theorem~\ref{thm:fdfa-to-nba}, we have an \nba\ of size smaller than $(2^{\Omega(n\log n)})^5=2^{\Omega(n\log n)}$ for $L_n^c$. Contradiction.
\end{proof}

\subsection{From FDFAs to $\omega$-automata}
We show that saturated \fdfa s can be polynomially translated into \nba s, yet translations of saturated \fdfa s to \dpa s may involve an inevitable exponential blowup.

\subsubsection*{To nondeterministic $\omega$-automata}
We show below that every saturated \fdfa\ can be polynomially translated to an equivalent \nba. Since an \nba\ can be viewed as a special case of an \npa, a translation of saturated \fdfa s to \npa s follows. Translating saturated \fdfa s to \nca s is not always possible, as the latter are not expressive enough.

The translation goes along the lines of the construction given in \cite{CalbrixNP93} for translating an $\ldollar$-automaton into an equivalent \nba. We prove below that the construction is correct for saturated \fdfa s, despite the fact that saturated \fdfa s can be exponentially smaller than $\ldollar$-automata. 

We start with a lemma from \cite{CalbrixNP93}, which will serve us for one direction of the proof.

\begin{lem}[\cite{CalbrixNP93}]\label{lem:MN}
Let $M,N\subseteq\Sigma^*$ such that $M\con N^* = M$ and $N^+ = N$. Then for every  ultimately periodic word $w\in\Sigma^\omega$ we have that $w \in M\con N^\omega$ iff there exist words $u\in M$ and $v\in N$ such that $uv^\omega = w$.
\end{lem}

We continue with the translation and its correctness.

\begin{thm}\label{thm:fdfa-to-nba}
For every saturated \fdfa\ $\F$ of size $(n,k)$, there exists an \nba\ $\B$ with $O(n^2 k^3)$ states, such that $\Lang{\F}\Characterizes\Lang{\B}$.
\end{thm}

\pagebreak
\begin{proof}\

\subproof{Construction}
Consider a saturated \fdfa\ $\F = (\Q, \PP)$, where $\Q = \la \Sigma, Q, \initstate, \delta \ra$, and for each state $q\in Q$, $\PP$ has the progress \dfa\ $\P_q = \la \Sigma, P_q, \initstate_q, \delta_q, F_q \ra$.

For every $q\in Q$, let $M_q$ be the language of finite words on which $\Q$ reaches $q$, namely $M_q = \{u \in \Sigma^* \ST \Q(u)=q\}$. 
For every $q\in Q$ and for every accepting state $f\in F_q$, let $N_{q,f}$ be the language of finite words on which $\Q$ makes a self-loop on $q$, $\P_q$ reaches $f$, and $\P_q$ makes a self-loop on $f$, namely $N_{q,f} = \{v \in \Sigma^* \ST (\delta(q,v)=q) \land (\P_q(v)=f) \land (\delta_q(f,v)=f)  \}$.
We define the $\omega$-regular language 
\begin{equation}\label{eq:L}
L = \bigcup_{  \{ (q,f) \ST (q\in Q) \land (f \in F_q)   \}  } M_q \con N_{q,f}^\omega
\end{equation}

One can construct an \nba\ $\B$ that recognizes $L$ and has up to $O(n^2 k^3)$ states:
$L$ is the union of $nk$ sublanguages; $\B$ will have $nk$ corresponding subautomata, and will nondeterministically start in one of them. In each subautomaton, recognizing the language $M_q \con N_{q,f}^\omega$, a component of size $n$ for $M_q$ is obtained by a small modification to $\Q$, in which $q$ can nondeterministically continue with an $\epsilon$-transition\footnote{The $\epsilon$-transitions can be removed from an \nba\ with no state blowup.} to a component realizing $N_{q,f}^\omega$. An \nba\ for the language $N_{q,f}$ consists of the intersection of three \nba s, for the languages $\{v \in \Sigma^* \ST \delta(q,v)=q  \}$, $\{v \in \Sigma^* \ST \P_q(v)=f \}$, and $\{v \in \Sigma^* \ST \delta_q(f,v)=f  \}$, each of which can be obtained by small modifications to either $\Q$ or $\P_q$, resulting in $nk^2$ states. Finally, the automaton for $N_{q,f}^\omega$ is obtained by adding $\epsilon$-transitions in the automaton of $N_{q,f}$ from its accepting states to its initial state. Thus, each subautomaton is of size $n+nk^2$, and $\B$ is of size $nk(n+nk^2)\in O(n^2k^3)$.

\subproof{Correctness}
Consider an ultimately periodic word $uv^\omega \in \Lang{\B}$. By the construction of $\B$, $uv^\omega \in L$, where $L$ is defined by Equation~(\ref{eq:L}). Hence, $uv^\omega \in M_q \con N_{q,f}^\omega$, for some $q\in Q$ and $f\in F_q$. By the definitions of $M_q$ and $N_{q,f}$, we get that $M_q$ and $N_{q,f}$ satisfy the hypothesis of Lemma~\ref{lem:MN}, namely $N_{q,f}^+ = N_{q,f}$ and $M_q \con N_{q,f}^* = M_q$. Therefore, by Lemma~\ref{lem:MN}, there exist finite words $u' \in M_q$ and $v' \in N_{q,f}$ such that $u'v'^\omega = uv^\omega$. From the definitions of $M_q$ and $N_{q,f}$, it follows that the run of $\Q$ on $u'$ ends in the state $q$, and $\P_q$ accepts $v'$. Furthermore, by the definition of $N_{q,f}$, we have $\delta(q,v')=q$, implying that $(u',v')$ is the normalization of itself. Hence, $(u',v')\in\Lang{\F}$. Since $\F$ is saturated and $u'v'^\omega = uv^\omega$, it follows that $(u,v)\in\Lang{\F}$, as required.

As for the other direction, consider a pair $(u,v) \in \Lang{\F}$, and let $(x,y)$ be the normalization of $(u,v)$ w.r.t.\ $\Q$. We will show that $xy^\omega\in L$, where $L$ is defined by Equation~(\ref{eq:L}), implying that $uv^\omega\in \Lang{\B}$. 
Let $q=\Q(x)$, so we have that $\P_q(y)$ reaches some accepting state $f$ of $\P_q$. Note, however, that it still does not guarantee that $y \in N_{q,f}$, since it might be that $\delta_q(f,y)\neq f$.

To prove that $xy^\omega\in L$, we will show that there is a pair $(x,y')\in\Sigmasp$ and an accepting state $f'\in\P_q$, such that $y' = y^t$ for some positive integer $t$, and $y'\in N_{q,f'}$; namely $\delta(q,y')= q$, $\P_q(y')= f'$, and $\delta_q(f',y')= f'$. Note first that since $\F$ is saturated, it follows that for every positive integer $i$, $(x,y^i)\in\Lang{\F}$, as $x(y^i)^\omega=xy^\omega$.

Now, for every positive integer $i$, $\P_q$ reaches some accepting state $f_i$ when running on $y^i$. Since $\P_q$ has finitely many states, for a large enough $i$, $\P_q$ must reach the same accepting state $\hat f$ twice when running on $y^i$. Let $h$ be the smallest positive integer such that $\P_q(y^h)=\hat f$, and $r$ the smallest positive integer such that $\delta_q(\hat f,y^r)=\hat f$. Now, one can verify that choosing $t$ to be an integer that is bigger than or equal to $h$ and is divisible by $r$ guarantees that $\delta(q,y^t)= q$ and $\delta_q(f',y^t)= f'$, where $f'=\P_q(y^t)$.
\end{proof}

\subsubsection*{To deterministic $\omega$-automata}
Deterministic \buchi\ and \cobuchi\ automata are not expressive enough for recognizing every $\omega$-regular language. We thus consider the translation of saturated \fdfa s to deterministic parity automata. A translation is possible by first polynomially translating the \fdfa\ into an \nba\ (Theorem~\ref{thm:fdfa-to-nba}) and then determinizing the latter into a \dpa\ (which might involve a $2^{O(n\log n)}$ state blowup \cite{Michel88}).

\begin{prop}\label{prop:FdfaToDpa}
Let $\F$ be a saturated \fdfa\ of size $(n,k)$. There exists a \dpa\ $\D$ of size $2^{O(n^2 k^3 \log n^2 k^3)}$, such that $•\Lang{\F}\Characterizes\Lang{\D}$.
\end{prop}

We show below that an exponential state blowup is inevitable.\footnote{This is also the case when translating \fdfa s to deterministic Rabin \cite{Rab69} and Streett \cite{Str82} automata, as explained in Remark~\ref{rem:FdfaToRabinAndStreett}.}
The family of languages $\{L_n\}_{n\geq 1}$ below demonstrates the inherent gap between \fdfa s and \dpa s; an \fdfa\ for $L_n$ may only ``remember'' the smallest and biggest read numbers among $\{1,2,...,n\}$, using $n^2$ states, while a \dpa\ for it must have at least $2^{n-1}$ states.

We define the family of languages $\{L_n\}_{n\geq 1}$ as follows. 
The alphabet of $L_n$ is $\{1,2,...,n\}$, and a word belongs to it iff the following two conditions are met:
\begin{itemize}
\item A letter $i$ is always followed by a letter $j$, such that $j \leq i+1$. For example, $533245\ldots$ is a bad prefix, since 2 was followed by 4, while $55234122\ldots$ is a good prefix.
\item The number of letters that appear infinitely often is odd. For example, $2331(22343233)^\omega$ is in $L_n$, while $1(233)^\omega$ is not.
\end{itemize}
We show below how to construct, for every $n\geq 1$, a saturated \fdfa\ of size polynomial in $n$ that characterizes $L_n$. Intuitively, the leading automaton handles the safety condition of $L_n$, having $n+1$ states, and ensuring that a letter $i$ is always followed by a letter $j$, such that $j \leq i+1$. The progress automata, which are identical, maintain the smallest and biggest number-letters that appeared, denoted by $s$ and $b$, respectively. Since a number-letter $i$ cannot be followed by a number-letter $j$, such that $j>i+1$, it follows that the total number of letters that appeared is equal to $b-s+1$. Then, a state is accepting iff $b-s+1$ is odd.

\begin{lem}\label{lem:FdfaForLn}
Let $n\geq 1$. There is a saturated \fdfa\ of size  $(n+1,n^2)$ characterizing $L_n$.
\end{lem}
\begin{proof}
We formally define an \fdfa\ $\F=(\Q, \PP)$ for $L_n$ over $\Sigma=\{1,2,\ldots,n\}$, as follows.

The leading automaton is $\Q=(\Sigma,Q,\initstate,\delta)$, where $Q=\{\bot, q_1, q_2, \ldots, q_n\}$; $\initstate = q_n$; and for every $i,j\in[1..n], \delta(q_i,j) = q_j$ if $j \leq i+1$, and $\bot$ otherwise, and $\delta(\bot,j)=\bot$.

The progress automaton for the state $\bot$ consists of a single non-accepting state with a self-loop over all letters.

For every $i\in[1..n]$, the progress automaton for $q_i$ is $\P_i =(\Sigma,P_i,\initstate_i,\delta_i, F_i)$, where:
\begin{itemize}
\item $P_i = [1..n]\times [1..n]$
\item $\initstate_i = (n, 1)$
\item $\delta_i$: For every $\sigma\in\Sigma$ and $s,b\in[1..n], \delta_i((s,b),\sigma)= (\min(s,\sigma) , \max(b, \sigma))$.
\item $F_i = \{  (s,b) \ST b-s \mbox{ is even }\}$
\end{itemize}
Notice that the progress automaton need not handle the safety requirement, as the leading automaton ensures it, due to the normalization in the acceptance criterion of an \fdfa. \end{proof}

A \dpa\ for $L_n$ cannot just remember the smallest and largest letters that were read, as these letters might not appear infinitely often. Furthermore, we prove below that the \dpa\ must be of size exponential in $n$, by showing that its state space must be doubled when moving from $L_n$ to $L_{n+1}$.

\begin{lem}\label{lem:DpaForLn}
Every \dpa\ that recognizes $L_n$ must have at least $2^{n-1}$ states.
\end{lem}
\begin{proof}
\newcommand{\LL}{\mathfrak{L}}
The basic idea behind the proof is that the \dpa\ cannot mix between 2 cycles of $n$ different letters each. This is because a mixed cycle in a parity automaton is accepting/rejecting if its two sub-cycles are, while according to the definition of $L_n$, a mixed cycle might reject even though both of its sub-cycles accept, and vice versa. Hence, whenever adding a letter, the state space must be doubled.

In the formal proof below, we dub a reachable state from which the automaton can accept some word a \emph{live state}, and for every $n\in\NN\setminus\{0\}$, define the alphabet $\Sigma_n=\{1, 2, \ldots, n\}$. Consider a \dpa\ $\D_n$ over $\Sigma_n$ that recognizes $L_n$, and let $q$ be some live state of $\D_n$. Observe that $\Lang{\D_n^q}$, namely the language of the automaton that we get from $\D_n$ by changing the initial state to $q$, is the same as $L_n$ except for having some restriction on the word prefixes. 
More formally, for every $n\in\NN$ and $u\in \Sigma_n^*$, we define the language $L_{n,u}=\{w~|~uw\in L_n\}$, and let $\LL_n$ denote the set of languages $\{ L_{n,u}\neq\emptyset~|~u\in\Sigma_n^*\}$. Given some $L_{n,u}$, since $L_{n,u}\neq\emptyset$, there is a live state $q$ that $\D_n$ reaches when reading $u$, and we have $L_{n,u}=\Lang{\D_n^q}$. Conversely, given a \dpa\ $D_n$ for $L_n$ and a live state $q$ of $\D_n$, let $u$ be a finite word on which $\D_n$ reaches $q$, then  $\Lang{\D_n^q}=L_{n,u}$. 

We prove by induction on $n$ the following claim, from which the statement of the lemma immediately follows: Let $\D_n$ be a \dpa\ over $\Sigma_n$ that recognizes some language in $\LL_n$. Then there are finite words $u,v \in\Sigma_n^*$, such that: 
\begin{enumerate}[label=\roman*)]
\item $v$ contains all the letters in $\Sigma_n$;  
\item the run of $\D_n$ on $u$ reaches some live state $p$; and 
\item the run of $\D_n$ on $v$ from $p$ returns to $p$, while visiting at least $2^{n-1}$ different  states.
\end{enumerate}
The base cases, for $n\in\{1,2\}$, are simple, as they mean a cycle of size at least $1$ over $v$, for $n=1$, and at least $2$ for $n=2$. Formally, for $n=1$, $L_1=\{1^\omega\}$ and $\LL_1=\{L_1\}$. A \dpa\ $\D_1$ for $L_1$ must have a live state $p$ that is visited infinitely often when $\D_1$ reads $1^\omega$. Thus, there is a cycle from $p$ back to $p$ along a word $v$ of length at least 1, containing the letter `1'. For $n=2$, it is still the case that the restriction on the next letter (to be up to 1-bigger than the current letter) does not influence. Hence, $L_2=\{u1^\omega \ST u\in\Sigma_2^*\} \cup \{u2^\omega \ST u\in\Sigma_2^*\}$ and $\LL_2=\{L_2\}$. Consider a \dpa\ $\D_2$ for $L_2$, having $k$ states. The run of $\D_2$ over the finite word $(12)^k$ must reach some live state $p$ twice. Thus, there is a cycle from $p$ back to itself along a word $v$ that contains both `1' and `2'. Observe that $v$ must be of length at least 2, as otherwise there are only self loops from $p$, implying that $\D_2^p$ either accepts or rejects all words, in contradiction to the definition of $\D_2$. 

In the induction step, we consider a \dpa\ $\D_{n+1}$, for $n\geq 2$, that recognizes some language $L \in \LL_{n+1}$. We shall define $\D'$ and $\D''$ to be the \dpa s that result from $\D_{n+1}$ by removing all the transitions over the letter $n\mathord{+}1$ and by removing all the transitions over the letter $1$, respectively. 

Observe that for every state $q$ that is live w.r.t.\ $\D_{n+1}$, we have that $\Lang{\D'^q}\in \LL_n$, namely the language of the \dpa\ that results from $\D_{n+1}$ by removing all the transitions over the letter $n\mathord{+}1$ and setting the initial state to $q$ is in $\LL_n$. This is the case since even if $q$ is only reachable via the letter $n\mathord{+}1$, it must have outgoing transitions over letters in $[2..n]$. Analogously, $\Lang{\D''^q}$ is isomorphic to a language in $\LL_n$ via the alphabet mapping of $i \mapsto (i-1)$. Hence, for every state $q$ that is live w.r.t.\ $\D_{n+1}$, the induction hypothesis holds for $\D'^q$ and $\D''^q$.

We shall prove the induction claim by describing words $u,v\in\Sigma_{n+1}^*$, and showing that they satisfy the requirements above w.r.t.\ $\D_{n+1}$.
We iteratively construct words $u'_i, u''_i, v'_i, v''_i$ until, roughly speaking, the run of $\D_{n+1}$ on the concatenation of these words closes a loop. 
More precisely, upon the first iteration $k$ for which there exists $j<k$ such that $\D_{n+1}(u'_1 u''_1 u'_2 u''_2 \ldots u'_j u''_j) = \D_{n+1}(u'_1\, u''_1\, u'_2\, u''_2 \ldots u'_k\, u''_k)$, we stop the iteration and 
define $u$ to be the word $u'_1 u''_1 u'_2 u''_2 \ldots u'_j u''_j$ and $v$ to be the word $u'_{j+1}\, v'_{j+1}\, u''_{j+1}\, v''_{j+1} \ldots u'_k\, v'_k\, u''_k\, v''_k$.

\vspace{2mm}
\noindent For the first iteration, we define:
\begin{itemize}
\item $u'_1$ and $v'_1$ are the words that follow from the induction hypothesis on $\D'^{q_1}$, where $q_1$ is the initial state of $\D_{n+1}$.
\item $u''_1$ and $v''_1$ are the words that follow from the induction hypothesis on $\D''^{q'_1}$, where $q'_1$ is the state that $\D_{n+1}$ reaches when reading $u'_1$.
\end{itemize}

For the next iterations, we define for every $i > 1$: 
\begin{itemize}
\item $u'_i$ and $v'_i$ are the words that follow from the induction hypothesis on $\D'^{q_i}$, where $q_i$ is the state that $\D_{n+1}$ reaches when reading $u'_1\, u''_1\, \ldots  u'_{i-1}\, u''_{i-1}$. (The state $q_i$ indeed belongs to $\D'^{q_i}$, since it has outgoing transitions on some letters in $[1..n]$.)
\item $u''_i$ and $v''_i$ are the words that follow from the induction hypothesis on $\D''^{q'_i}$, where $q'_i$ is the state that $\D_{n+1}$ reaches when reading $u'_1\, u''_1\, \ldots  u'_{i-1}\, u''_{i-1} \, u'_i$. (The state $q'_i$ indeed belongs to $\D''^{q_i}$, since it has outgoing transitions on some letters in $[2..n\mathord{+}1]$.)
\end{itemize}
Note that, for all $i$, by the induction hypothesis $v'_i$ contains all letters of $\Sigma_n$ and $v''_i$ 
contains all the letters in $\Sigma_{n+1}\setminus\{1\}$. Since $v$ is composed of $v'_i$'s and $v''_i$'s it follows that $v$ contains all the letters in $\Sigma_{n+1}$.  By the definition of $u$ and $v$, we also have that the run of $\D_{n+1}$ on $u$ reaches some live state $p$, and the run of $\D_{n+1}$ on $v$ from $p$ returns to $p$. 
Moreover, for every prefix  $v_1$ of $v$ that reaches  a state $s$, we have that $s$ is a live state, and for $v_2$ such that $v=v_1v_2$, the run of $\D_{n+1}$ on  $v_2v_1$  returns to $s$. 
Now, we need to prove that the run of $\D_{n+1}$ on $v$ from $p$ visits at least $2^{n}$ states.

\begin{figure}
	\begin{center}
		\scalebox{.99}{
			\begin{tikzpicture}
			\node (wbegin) {};
			\node [right=2.5cm of wbegin] (wmid) {};
			\node [right=10.4cm of wmid] (wend) {};
			\node [right=1.1cm of wmid] (vpribegin) {};
			\node [right=1.8cm of wmid] (vprimid) {};
			\node [right=2.9cm of wmid] (vpriend) {};
			\node [right=5.0cm of wmid] (vdpribegin) {};
			\node [right=6.05cm of wmid] (vdprimid) {};			
			\node [right=6.9cm of wmid] (vdpriend) {};
			
			\node [below=.3cm of wbegin] (wbegins) {};
			\node [below=.3cm of wmid] (wmids) {};
			\node [below=.3cm of wend] (wends) {};
			
			\node [above=.54cm of wmid] (wmidn) {};
			\node [above=.54cm of vprimid] (vprimidn) {};
			
			\node [above=1.1cm of wmid] (wmidnn) {};
			\node [above=1.1cm of vdprimid] (vdprimidnn) {};
			
			\node [above=1.4cm of wmid] (wmidnnn) {};
			\node [above=1.4cm of vdprimid] (vdprimidnnn) {};

			\draw (wbegin.center) -- (wend.center);
			\draw (wbegin.center) -- (wend.center);
			\draw (wbegin.north) -- (wbegin.south);
			\draw (wmid.north) -- (wmid.south);
			\draw (wend.north) -- (wend.south);
			\draw (vpribegin.north) -- (vpribegin.south);
			\draw (vprimid.north) -- (vprimid.south);
			\draw (vdprimid.north) -- (vdprimid.south);
			\draw (vpriend.north) -- (vpriend.south);
			\draw (vdpribegin.north) -- (vdpribegin.south);
			\draw (vdpriend.north) -- (vdpriend.south);
			
			\draw [dashed] (wmidnn.north) -- (wmids.south);
			\draw [dashed] (vprimidn.north) -- (vprimid.south);
			\draw [dashed] (vdprimidnn.north) -- (vdprimid.south);

			\draw [decorate, decoration={brace,amplitude=8pt}]
			(vpribegin.north) -- (vpriend.north)
			node [midway, above, yshift=8pt,xshift=1pt] (wlabel) {$v'$};
			
			\draw [decorate, decoration={brace,amplitude=8pt}]
			(vdpribegin.north) -- (vdpriend.north)
			node [midway, above, yshift=8pt, xshift=2.5pt] (wlabel) {$v''$};
			
			\draw [decorate, decoration={brace, mirror,amplitude=5pt}]
			(vpribegin.south) -- (vprimid.south)
			node [midway, below, yshift=-4pt] (wlabel) {$l'$};
			
			\draw [decorate, decoration={brace, mirror,amplitude=5pt}]
			(vprimid.south) -- (vpriend.south)
			node [midway, below, yshift=-4pt] (wlabel) {$r'$};
			
			\draw [decorate, decoration={brace, mirror,amplitude=5pt}]
			(vdpribegin.south) -- (vdprimid.south)
			node [midway, below, yshift=-4pt] (wlabel) {$l''$};
			
			\draw [decorate, decoration={brace, mirror,amplitude=5pt}]
			(vdprimid.south) -- (vdpriend.south)
			node [midway, below, yshift=-4pt] (wlabel) {$r''$};
			
			\draw [decorate, decoration={brace, mirror,amplitude=12pt}]
			(wbegins.south) -- (wmids.south)
			node [midway, below, yshift=-10pt] (wlabel) {$u$};
			
			\draw [decorate, decoration={brace, mirror,amplitude=12pt}]
			(wmids.south) -- (wends.south)
			node [midway, below, yshift=-10pt] (wlabel) {$v$};
			
			\draw [decorate, decoration={brace,amplitude=10pt}]
			(wmidn.north) -- (vprimidn.north)
			node [midway, above, yshift=8pt] (wlabel) {$x$};
			
			\draw [decorate, decoration={brace,amplitude=10pt}]
			(wmidnnn.south) -- (vdprimidnnn.south)
			node [midway, above, yshift=8pt] (wlabel) {$y$};

			\end{tikzpicture}
		}
	\end{center}
	\caption{Subwords of $v$.}\label{fig:subwords}
\end{figure}
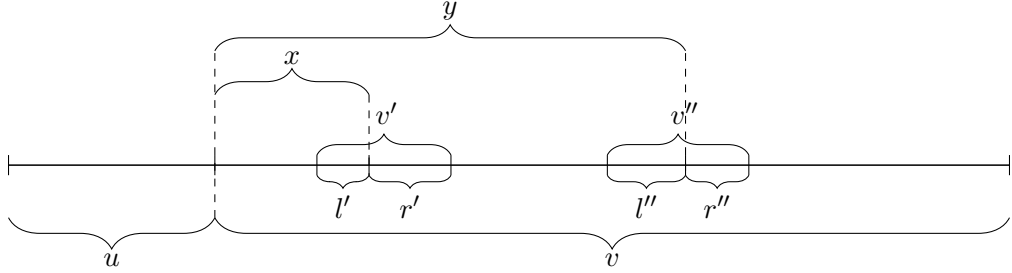

Let $v'$ be some $v'_i$ subword of $v$ and likewise let $v''$ be some $v''_j$ subword of $v$. 
We claim that the run of $\D_{n+1}$ on $uv^\omega$ traverses disjoint sets of states while reading the subwords $v'$ and $v''$. This will provide the required result, since when traversing either $v'$ or $v''$ we know that $\D_{n+1}$  visits at least $2^{n-1}$ different states, by the induction hypothesis.

Assume, by way of contradiction, a state $s$ that is visited by $\D_{n+1}$ when traversing both $v'$ and $v''$.
Let $l'$ and $r'$ be the parts of $v'$ that $\D_{n+1}$ reads before and after reaching $s$, respectively, and $l''$ and $r''$ the analogous parts of $v''$, as shown in Fig.~\ref{fig:subwords}. Let $x$ be the subword of $uv^\omega$ between $u$ and $r'$ and let $y$ be the subword of $uv^\omega$ between $u$ and $r''$.
Now, define the $\omega$-words $m'=u x  (r' \, l')^\omega$, $m''=u y \, (r'' \, l'')^\omega$, and $m=uy \, (r ''\,l'' \, r' \, l')^\omega$. 

Observe that $m'$ and $m''$ both belong or both do not belong to $L$, since there is the same number of letters ($n$) that appear infinitely often in each of them. The word $m$, on the other hand, belongs to $L$ if $m'$ and $m''$ do not belong to $L$, and vice versa, since $n+1$ letters appear infinitely often in it. However, the set of states that are visited infinitely often in the run of $\D_{n+1}$ on $m$ is the union of the sets of states that appear infinitely often in the runs of $\D_{n+1}$ on $m'$ and $m''$. Thus, if $\D_{n+1}$ accepts both $m'$ and $m''$ it also accepts $m$, and if it rejects both $m'$ and $m''$ it rejects $m$.  (This follows from the fact that the minimal number in a union of two sets is even/odd iff the minimum within both sets is even/odd.) Contradiction.
\end{proof}

\begin{thm}
There is a family of languages $\{L_n\}_{n\geq 1}$ over the alphabet $\{1,2,\ldots,n\}$, such that for every $n\geq 1$, there is a saturated \fdfa\ of size  $(n+1,n^2)$ that characterizes $L_n$, while a \dpa\ for $L_n$ must be of size at least $2^{n-1}$.
\end{thm}
\begin{proof}
By Lemmas \ref{lem:FdfaForLn} and \ref{lem:DpaForLn}.
\end{proof}

\begin{rem}\label{rem:FdfaToRabinAndStreett}
A small adaptation to the proof of Lemma~\ref{lem:DpaForLn} shows an inevitable exponential state blowup also when translating a saturated \fdfa\ to a deterministic $\omega$-automaton with a stronger acceptance condition of Rabin \cite{Rab69} or Streett \cite{Str82}: A mixed cycle in a Rabin automaton is rejecting if its two sub-cycles are, and a mixed cycle in a Streett automaton is accepting if its two sub-cycles are. Hence, the proof of Lemma~\ref{lem:DpaForLn} holds for both Rabin and Streett automata if proceeding in the induction step from an alphabet of size $n$ to an alphabet of size $n+2$, yielding a Rabin/Streett automaton of size at least $2^{\frac{n}{2}}$.

As for translating a saturated \fdfa\ to a deterministic Muller automaton \cite{Mul63}, it is known that translating a \dpa\ of size $n$ into a deterministic Muller automaton might require the latter to have an accepting set of size exponential in $n$ \cite{Saf89,Bok17a}. Hence, by Theorem~\ref{thm:dpa-to-fdfa}, which shows a polynomial translation of \dpa s to \fdfa s, we get that translating an \fdfa\ to a deterministic Muller automaton entails an accepting set of exponential size, in the worst case.
\end{rem}

\section{Discussion}
The interest in \fdfa s as a representation for $\omega$-regular languages stems from the fact that they possess a correlation between the automaton states and the language right congruences, a property that traditional $\omega$-automata lack. This property is beneficial in the context of learning, and indeed algorithms for learning $\omega$-regular languages by means of saturated \fdfa s were recently provided \cite{AngluinF14,LCZL16}. 
Analyzing the succinctness of saturated \fdfa s and the complexity of their Boolean operations and decision problems, we believe that they provide an interesting formalism for representing $\omega$-regular languages. Indeed, Boolean operations and decision problems can be performed in nondeterministic logarithmic space and their succinctness lies between deterministic and nondeterministic $\omega$-automata.

\bibliographystyle{alpha}% The preferred bibliography style
\bibliography{bib}
\end{document}

%% file: macros.tex
\newcommand{\commentout}[1]{}

\newcommand{\la}{\langle}
\newcommand{\ra}{\rangle}

\newcommand{\Times}{\!\times\!}

\newcommand{\Lang}[1]{{\llbracket #1 \rrbracket}}
\newcommand{\sema}[1]{{\Lang{#1}}}

\newcommand{\name}[1]{{\textsc{#1}}}
\newcommand{\fdfa}{\name{fdfa}}

\newcommand{\dfa}{\name{dfa}}
\newcommand{\nfa}{\name{nfa}}
\newcommand{\dba}{\name{dba}}
\newcommand{\nba}{\name{nba}}
\newcommand{\dca}{\name{dca}}
\newcommand{\nca}{\name{nca}}

\newcommand{\dpa}{\name{dpa}}
\newcommand{\npa}{\name{npa}}

\newcommand{\buchi}{B$\ddot{\textrm{u}}$chi}

\newcommand{\cobuchi}{co-B$\ddot{\textrm{u}}$chi}

\newcommand{\NN}{\mathbb{N}}

\newcommand{\ldollar}{\ensuremath{L_{\$}}}

\newcommand{\A}{{\mathcal{A}}}
\newcommand{\U}{{\mathcal{U}}}
\newcommand{\F}{{{\mathcal{F}}}}

\newcommand{\B}{{{\mathcal{B}}}}

\newcommand{\D}{{{\mathcal{D}}}}

\newcommand{\Q}{{{\mathcal{Q}}}}

\newcommand{\X}{{{\mathcal{X}}}}
\renewcommand{\H}{{{\mathcal{H}}}}
\renewcommand{\S}{{{\mathcal{S}}}}
\renewcommand{\P}{{{\mathcal{P}}}}

\newcommand{\PP}{{{\mathbf{P}}}}

\newcommand{\initstate}{\iota}
\newcommand{\Sigmasp}{{\Sigma^{*+}}}

\newcommand{\congr}[1]{\approx_{#1}}
\newcommand{\Characterizes}{\equiv}%{\approx}
\newcommand{\con}{\cdot}
\newcommand{\ST}{~|~}

\newcommand{\subproof}[1]{\paragraph*{\textit{#1}}}